\documentclass[11pt]{article}

\usepackage{fullpage}
\usepackage{latexsym}
\usepackage{amssymb,amsfonts,amsmath,amsthm}
\usepackage{mathpazo}
\usepackage{epsfig}

\def\Full{\rho}
\def\Residue{\sigma}
\def\pure{\phi}
\def\opure{\widetilde{\pure}}

\def\OO{{O_{p}^{k}}}
\def\H{H^{k}}

\newcommand{\FState}[1]{\Full^{k}_{#1}}
\newcommand{\RState}[1]{\Residue^{k}_{#1}}
\newcommand{\pState}[1]{\pure^{k}_{#1}}
\newcommand{\uState}[1]{\pure^{0}_{#1}}

\newcommand{\amp}[1]{\alpha^{k}_{#1}}
\newcommand{\uamp}[1]{\alpha^{0}_{#1}}

\newcommand{\U}[1]{U_{#1}}

\newcommand{\OFState}[1]{\widetilde{\Full}^{k}_{#1}}
\newcommand{\OpState}[1]{\widetilde{\pure}^{k}_{#1}}
\def\01{\{0,1\}}

\newcommand{\ket}[1]{|#1\rangle}
\newcommand{\bra}[1]{\langle#1|}
\newcommand{\ketbra}[2]{|#1\rangle\langle#2|}
\newcommand{\norm}[1]{{\left\|{#1}\right\|}}
\newcommand{\normb}[1]{{\big\|{#1}\big\|}}
\newcommand{\inp}[2]{\langle{#1}|{#2}\rangle} % inproduct, < | >
\newcommand{\C}{\mathbb{C}}

\newcommand{\trnorm}[1]{\norm{#1}_{\rm tr}}

\newtheorem{definition}{Definition}
\newtheorem{theorem}{Theorem}
\newtheorem{lemma}[theorem]{Lemma}

\newtheorem{claim}[theorem]{Claim}

\renewcommand{\qed}{\hfill{\rule{2mm}{2mm}}}
\renewenvironment{proof}[1][]{\begin{trivlist}
\item[\hspace{\labelsep}{\bf\noindent Proof#1:\/}] }{\qed\end{trivlist}}

\begin{document}

\title{Impossibility of a Quantum Speed-up with a Faulty Oracle}
\author{
Oded Regev\thanks{School of Computer Science, Tel-Aviv University, Tel-Aviv 69978, Israel. Supported
   by the Binational Science Foundation, by the Israel Science Foundation, and
   by the European Commission under the Integrated Project QAP funded by the IST directorate as Contract Number 015848.
}
\and
Liron Schiff\thanks{School of Computer Science, Tel-Aviv University, Tel-Aviv 69978, Israel.}
}
\date{}
\maketitle

\begin{abstract}
We consider Grover's unstructured search problem in the setting where each oracle call has some small probability of
failing. We show that no quantum speed-up is possible in this case.
\end{abstract}

\section{Introduction}

\paragraph{Unstructured search problem:}
The unstructured search problem, also known as the unordered search problem or as Grover's search problem, is the most basic problem in the query model. The goal is to find a marked
entry out of $N$ possible entries. In this model the entries are accessible only through a black box
(the oracle), and the complexity of the algorithm is measured in terms of the number of oracle queries.
In the classical world, it is easy to see that solving this search problem requires $\Theta(N)$
queries, even if we allow randomization. In the quantum world, however, one can find a marked
item with only $O(\sqrt{N})$ queries, as was shown in Grover's seminal paper \cite{Grover96}.
Moreover, it is known that this is optimal (see, e.g.,~\cite{BennettBBV97,BoyerBHT98,Ambainis00}).
This remarkable quadratic improvement is considered one of the biggest successes of quantum computing,
and has sparked a huge interest in the quantum query model (see \cite{Ambainis05} for a recent survey).

\paragraph{Searching with a faulty oracle:}
In this paper we consider the unstructured search problem in the {\em faulty oracle model}, a question
originally presented to us by Harrow~\cite{Harrow06}.
In this model, each oracle call succeeds with some probability
$1-p$, and with the remaining probability $p$ the state given to the oracle remains unchanged. More formally,
each oracle call maps an input state $\rho$ into $(1-p)\cdot O \rho O^\dag + p\cdot \rho$ where $O$ is the original
(unitary) oracle operation.
We note that this model can be seen to be equivalent to other, seemingly more realistic, models of faults,
such as the model considered in Shenvi et al. \cite{ShenviBW03} in which
the oracle's operation is subject to small random phase fluctuations.

Our motivation for considering the faulty oracle model is twofold.
First, we believe that since the unstructured search
problem is such a basic question, it is theoretically interesting to consider it in different settings,
as this might shed more light on the strengths and weaknesses of quantum query algorithms.
A second motivation is related to implementation aspects of quantum query algorithms,
as one can expect any future implementation of a Grover oracle to be imperfect
(see \cite{ShenviBW03} for a further discussion of the physical significance of the model).

To motivate our main result and to get some intuition for the model, let us consider the behavior of Grover's original
algorithm in this setting. Recall that Grover's algorithm can be seen as a sequence of two alternating reflections,
$OUOUOU\cdots OU$ where $U$ is the reflection given by Grover's algorithm and $O$ is the reflection representing the
oracle call. In the analysis of Grover's algorithm, one observes that the state of the system is restricted to a
two-dimensional subspace, inside which lie the initial state and the target state. The angle between these two states
is essentially $\pi/2$. Furthermore, the combined operation $OU$ of two consecutive reflections can be seen a rotation
by an angle of essentially $1/\sqrt{N}$ inside this two dimensional subspace. Hence the total number of oracle calls
required to get to the target state is $O(\sqrt{N})$.

In the faulty oracle model, each oracle call $O$ has some constant probability of not doing anything. Hence, the
sequence of reflections might look like $OUOUOUUOUOUOUO$. The effect of this is that after a sequence of rotations $OU$
by $1/\sqrt{N}$, we instead obtain a sequence of rotations $UO=(OU)^\dag$ by $-1/\sqrt{N}$ which cancel the previous
ones. The cancellation can also be seen by noting that $U^2=O^2=I$. The end result is that instead of rotating towards
the target, our rotation behaves like a random walk, alternating between steps of $1/\sqrt{N}$ and steps of
$-1/\sqrt{N}$. Using known properties of random walks on a line, the number of steps required for this walk to reach
the target is $\Theta(N)$, which shows that Grover's algorithm is no better than the naive classical search algorithm.

But can there be another, more sophisticated algorithm that copes better with the faults?
Our main result shows that the answer is essentially `no'.

\paragraph{Our result:}

Our main result shows that there is essentially no quantum advantage when searching with
a faulty oracle.

\begin{theorem}\label{thm:main-thm}
Any algorithm that solves the $p$-faulty Grover problem must use $T > \frac{p}{10(1-p)} N$ queries.
\end{theorem}

\noindent
In particular, for any constant $p>0$, this gives a lower bound of $\Omega(N)$.

Notice that the above statement holds for {\em any} quantum algorithm, and
not just for Grover's algorithm. In particular, it shows that some natural approaches,
like fault-tolerant quantum computation \cite{KnillZurekLaflamme},
cannot help in this setting. Note, however, that this impossibility result applies
only in case that the oracle is truly a black-box oracle; if, instead, the oracle is
given as a faulty circuit, then fault-tolerant schemes \emph{can} be used
to achieve a quantum speed-up by applying them to the
circuit obtained by taking Grover's algorithm and replacing the oracle calls with
their circuit implementation.

\paragraph{Related work:}
There has been a considerable amount of work dedicated to analyzing Grover's algorithm in
all kinds of faulty settings (see, e.g., \cite{LongLZT00,ShenviBW03,ShapiraMB03}).
All these works concentrate on Grover's algorithm (or variants thereof) and
none of them give a general statement that applies to all algorithms.
In particular, Shenvi et al. \cite{ShenviBW03} analyze the behavior of Grover's
algorithm in a physically motivated model that is
equivalent to ours. Our result answers the main open question presented in their paper.

There has also been a significant amount of work on searching with an imperfect, but still unitary, oracle (see, e.g., \cite{BrassardHMT02,HoyerMdW03,BuhrmanNRdW07,IwamaRY05,SuzukiYNW06}).
Such oracles are sometimes known as {\em noisy} oracles.
The motivation for this model is algorithmic, and is related to what is known as amplitude amplification. Typically
in this case, the quantum speed-up of $O(\sqrt{N})$ is still achievable. Very roughly speaking, this is
because a unitary operation (even an imperfect one) is reversible and does not lead to decoherence.
There has also been some recent work on analyzing the case of an imperfect unitary implementation of Grover's algorithm (as opposed to an imperfect \emph{oracle})
\cite{MagniezNRS07}, again showing that a speed-up of $O(\sqrt{N})$ is achievable.

\paragraph{Open problems:}
One interesting open question is to extend our result to other physically interesting fault models.
We believe that our proof technique should be applicable in a more general setting.
One natural fault model suggested to us by Nicolas Cerf is the one
in which each oracle query has probability $p$ of turning the state into
the completely mixed state.
Also, is there \emph{any} reasonable fault model for which a quantum speed-up \emph{is} achievable?
We suspect that the answer is no.

Another open question is to extend our result to other search problems
(see \cite{Ambainis05} for a recent survey).
Is there {\em any} search problem for which a quantum speed-up is achievable with a faulty oracle?
Can one extend our lower bound to a more general lower
bound in the spirit of the adversary method (see \cite{Ambainis00,HoyerLS07})?
It is also worth investigating whether the polynomial method \cite{BealsBCMdW01} can be used
to derive lower bounds in the faulty oracle case; our attempts to do so were unsuccessful.
We should emphasize, however, that our faulty oracle model is not necessarily
so natural for other search problems, and before approaching the above open questions,
some thought should be given to the choice of the faulty oracle model.

\section{Preliminaries}

We assume familiarity with basic notions of quantum computation (see \cite{Nielsen:book}).

\begin{definition}[Grover oracle]
For each $k\in\{1,\ldots,N\}$ where $N$ is an integer,
the \emph{perfect oracle} $\hat{O}^{k}$ is the unitary transformation acting on an $N$-dimensional register that maps $\ket{k}$ to $-\ket{k}$
and $\ket{i}$ to $\ket{i}$ for each $i \neq k$, i.e.,
$$\hat{O}^k = - \ketbra{k}{k} + \sum_{i \neq k} \ketbra{i}{i}.$$
We also extend the definition to $k=0$ by defining $\hat{O}^0$ to be the `null' oracle, given by the identity matrix $I$.
\end{definition}

\begin{definition}\label{def:faultyO}
The \emph{$p$-faulty oracle} $\OO$ is defined as the operation that with probability $1-p$,  acts as the perfect oracle $\hat{O}^{k}$ and otherwise does nothing, i.e., for any density matrix $\Full$,
$$\OO(\Full) = (1-p)\cdot \hat{O}^k \Full \hat{O}^{k \dag} + p\cdot\Full.$$
\end{definition}

\noindent
We note that instead of our phase-flipping oracle, one could also consider a bit-flipping oracle. Since
it is not difficult to construct the latter from the former  (see, e.g.,
\cite[Chapter 8]{KayeLM07}), our lower bound also applies to the bit-flipping case.

\begin{definition}
Let $0< p < 1$ be some constant.
In the \emph{$p$-faulty Grover problem}, we are given oracle access to the $p$-faulty oracle $\OO$ for some unknown $k \in \{0,\ldots,N\}$ and our goal is to decide whether $k=0$ or not with success probability at least $\frac{9}{10}$.
\end{definition}

\noindent
Note that the choice of success probability is inconsequential, as one can easily increase it by repeating
the algorithm a few times. Also note that we consider here the decision problem, as opposed to the search
problem of recovering $k$ from $\OO$. Since we are interested in lower bounds, this makes our result
stronger.

\section{Proof}

We start by giving a brief outline of the proof. For simplicity, we consider the case $p=1/2$.
The proof starts with a simple, yet crucial, observation (Claim~\ref{clm:O_pure}) which
gives an alternative description of the faulty oracle. In the case $p=1/2$, it says that
the oracle $\OO$ is essentially performing the two-outcome measurement
given by $\{\ket{k}, \ket{k}^\perp\}$. Then, in Lemma~\ref{lem:g_decompose}, we `approximate' the
mixed states that arise during the algorithm with (unnormalized) pure states. This is done by
assuming that the measurements done by the oracle all end up in the $\ket{k}^\perp$ subspace.
The rest of the proof is similar in structure to previous lower bounds. Using the pure state description,
we define a progress measure $\H_{t}$, which is initially zero. We show that at the end of the algorithm it must be high (Lemma~\ref{lem:F}),
and that it cannot increase by too much at each step (Lemma~\ref{lem:dH}). This yields the
desired lower bound on the number of queries $T$. We now proceed with the proof.

Let $A$ be an algorithm for the $p$-faulty Grover problem on $N$ elements that uses $T$ queries.
Assume the algorithm is described by the unitary operations $\U{0},\U{1},\U{2},\ldots,\U{T}$ acting
on an $NM$-dimensional system, composed of an $N$-dimensional query register used as oracle input, and an
$M$-dimensional ancillary register. Let $\widetilde{\Full}_0$ denote the initial state of the system, which we assume without
loss of generality to be a pure state $\widetilde{\Full}_0=\ket{\widetilde{\pure}_0}\bra{\widetilde{\pure}_0}$. For $k \in \{0,\ldots,N\}$,
we let $\OFState{0}=\widetilde{\Full}_0$, $\FState{0}=\U{0}\OFState{0} \U{0}^\dag$,
$\OFState{1}=\OO (\FState{0})$, $\FState{1} = \U{1}\OFState{1} \U{1}^\dag, \ldots,\FState{T} = \U{T}\OFState{T} \U{T}^\dag$ be the intermediate states of the algorithm
when run with oracle $\OO$  (see Figure~\ref{fig:progress}). In other words, $\FState{t}$ is the state of the system right after applying $\U{t}$,
and $\OFState{t+1}$ is the state of the system right after applying $\OO$ on $\FState{t}$.

\begin{figure}[h]
\center{
 \epsfxsize=5in
 \epsfbox{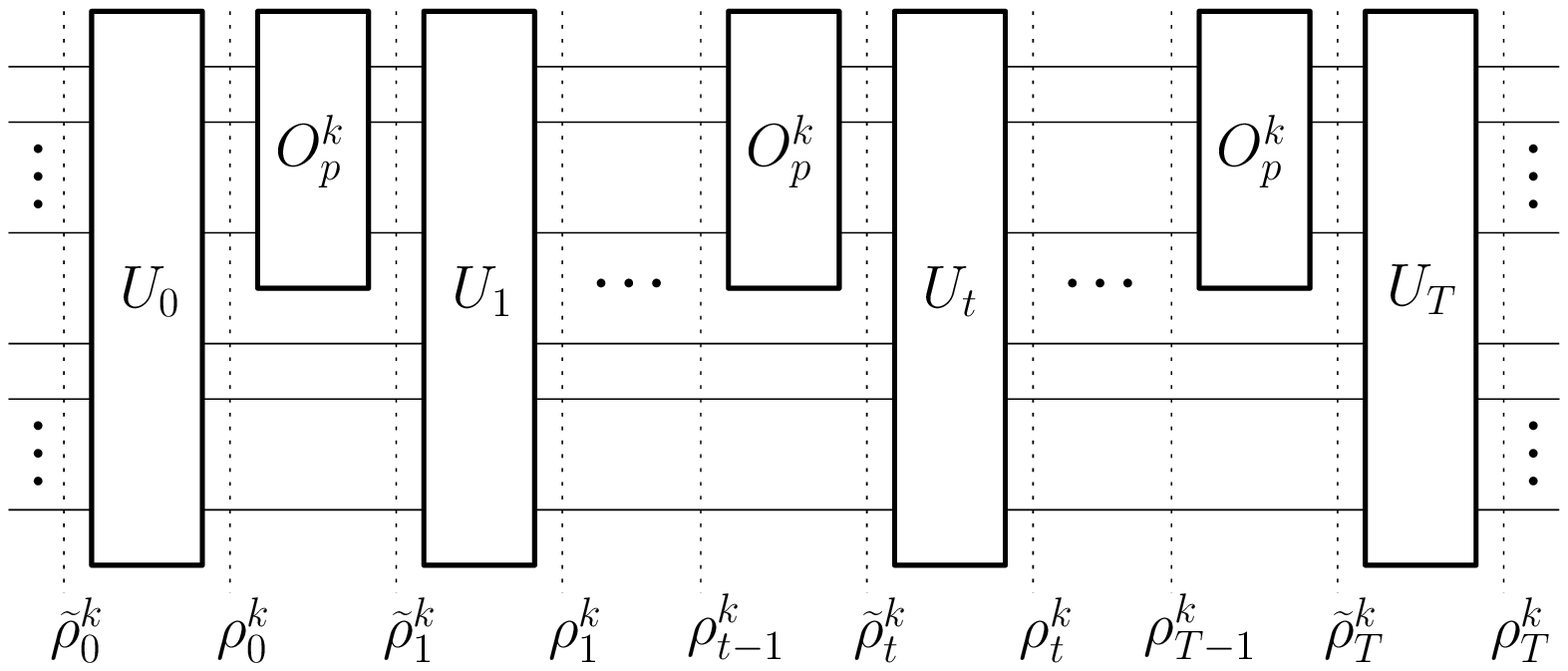}}
 \caption{State evolution.}
 \label{fig:progress}
\end{figure}

First we show a different way to decompose the outcome of $\OO$.
\begin{claim}\label{clm:O_pure}
Let $\ket{\pure} \in \C^{N\cdot M}$ be an arbitrary vector and let $\ket{\beta_i} \in \C^M$
be such that $\ket{\pure}=\sum_{i=1}^{N} \ket{i,\beta_i}$. Then
$$\OO(\ketbra{\pure}{\pure}) = \ketbra{\opure}{\opure}  + 4p(1-p)\ket{k,\beta_k}\bra{k,\beta_k}$$
    where
    $$\ket{\opure} := \sum_{i=1}^{N}\ket{i,\beta_i} - 2(1-p)\ket{k,\beta_k}.$$
\end{claim}

\begin{proof}
By Definition~\ref{def:faultyO} we have
$$\OO(\ketbra{\pure}{\pure}) = p \ketbra{\pure}{\pure} + (1-p)\ket{\psi}\bra{\psi}$$
where $\ket{\psi} = \sum_{i\neq k} \ket{i,\beta_i}  - \ket{k,\beta_k}$.
Therefore
\begin{align*}
\OO(\ketbra{\pure}{\pure}) = & \sum_{i\neq k} \sum_{j\neq k}\ket{i,\beta_i}\bra{j,\beta_j} - (1-2p)\sum_{j\neq k}\ket{k,\beta_k}\bra{j,\beta_j} - (1-2p)\sum_{i\neq k}\ket{i,\beta_i}\bra{k,\beta_k} + \ket{k,\beta_k}\bra{k,\beta_k}\\
=&  \left(\sum_{i\neq k}\ket{i,\beta_i} - (1-2p)\ket{k,\beta_k}\right) \left(\sum_{j\neq k}\bra{j,\beta_j} -(1-2p)\bra{k,\beta_k}\right)\\
    &+(1-(1-2p)^2)\ketbra{k,\beta_k}{k,\beta_k}.
\end{align*}
\end{proof}

We will use the following vectors to track the progress of the algorithm.
\begin{definition}\label{def:PState}
For $k \in \{0,\ldots,N\}$ and $t \in \{0,\ldots, T\}$ we define
the vectors $\ket{\pState{t}},\ket{\OpState{t}}\in \C^{N \cdot M}$ and $\ket{\amp{t,i}} \in \C^{M}$ as follows. First,
\begin{align*}
\ket{\OpState{0}} &:= \ket{\widetilde{\pure}_0}, \\
\ket{\pState{t}} &:= \U{t}\ket{\OpState{t}}
\end{align*}
and $\ket{\amp{t,i}}$ are given by
$$\ket{\pState{t}} = \sum_{i=1}^N \ket{i,\amp{t,i}}.$$
Finally, for $k \in \{1,\ldots,N\}$ and $t \in \{0,\ldots,T-1\}$ we define
$$\ket{\OpState{t+1}} :=  \ket{\pState{t}} - 2(1-p)\ket{k,\amp{t,k}} = \sum_{i=1}^N \ket{i,\amp{t,i}} - 2(1-p)\ket{k,\amp{t,k}}$$
and for $k=0$ we define $\ket{\widetilde{\pure}^{0}_{t+1}} :=  \ket{{\pure}^{0}_{t}}$.
\end{definition}

\begin{lemma}\label{lem:g_decompose}
For all $t \in \{0,\ldots,T\}$ and $k \in \{1,\ldots,N\}$, we can write
    $$ \FState{t} = \ket{\pState{t}}\bra{\pState{t}} + \RState{t}$$
for some positive semidefinite matrix $\RState{t}$.
\end{lemma}
\begin{proof}
Fix some $k \in \{1,\ldots,N\}$.
The lemma clearly holds for $t=0$ (with $\Residue^{k}_{0} = 0$). Suppose the lemma holds for $t$ and let us prove it for $t+1$.
By the induction hypothesis,
\begin{align}\label{eq:OFState}
\OFState{t+1}=\OO(\FState{t}) = \OO(\ket{\pState{t}}\bra{\pState{t}}) + \OO(\RState{t}).
\end{align}
By Claim~\ref{clm:O_pure} and the definition of $\ket{\OpState{t}}\bra{\OpState{t}}$
%We obtain that
$$\OO(\ket{\pState{t}}\bra{\pState{t}}) = \ket{\OpState{t+1}}\bra{\OpState{t+1}} + 4p(1-p)\ket{k,\amp{t,k}}\bra{k,\amp{t,k}}.$$
By combining this with Eq.~\eqref{eq:OFState} we get
$$\OFState{t+1} = \ket{\OpState{t+1}}\bra{\OpState{t+1}} + 4p(1-p)\ket{k,\amp{t,k}}\bra{k,\amp{t,k}} + \OO(\RState{t}).$$
We apply $\U{t+1}$ and obtain
\begin{align*}
\FState{t+1} &= \U{t+1}\OFState{t+1}\U{t+1}^\dag\\
&= \U{t+1}\ket{\OpState{t+1}}\bra{\OpState{t+1}}\U{t+1}^\dag + \U{t+1}\left(4p(1-p)\ket{k,\amp{t,k}}\bra{k,\amp{t,k}} + \OO(\RState{t})\right)\U{t+1}^\dag\\
&= \ket{\pState{t+1}}\bra{\pState{t+1}} + \U{t+1}\left(4p(1-p)\ket{k,\amp{t,k}}\bra{k,\amp{t,k}} + \OO(\RState{t})\right)\U{t+1}^\dag.
\end{align*}
The second term is clearly positive semidefinite, as required.
\end{proof}

We now define our progress measure $\H_{t}$.
\begin{definition}\label{def:H}
For $t \in \{0,\ldots,T\}$ and $k \in \{1,\ldots,N\}$ we define
$$\H_{t} := \normb {\ket{\uState{t}}-\ket{\pState{t}}}^2 .$$
\end{definition}

\noindent
Notice that $\H_0 = 0$.
The following lemma shows that at the end of the algorithm, the progress measure must be not too small.
Intuitively, this holds since if $\H_T$ is small, then $\ket{\pState{T}}$ is close to $\ket{\uState{T}}$
and since the latter is a unit vector, the former must be of norm close to $1$. This, in turn, implies
that $\FState{T}$ is close to $\ket{\pState{T}}\bra{\pState{T}}$, which is close to
$\ket{\uState{T}}\bra{\uState{T}} = \Full^{0}_T$ and thus the algorithm cannot distinguish between
$\FState{T}$ and $\Full^{0}_T$ in contrast to our assumption about the algorithm. We proceed with the formal proof.
\begin{lemma}\label{lem:F}
For all $k\in \{1,\ldots,N\}$, $\H_{T} > \frac{1}{10}$.
\end{lemma}

\begin{proof}
By our assumption on the correctness of the algorithm,
\begin{align*}
\frac{9}{10} \leq \trnorm{\FState{T}-\Full^{0}_T} &= \trnorm{\FState{T} - \ket{\uState{T}}\bra{\uState{T}} } \\
&\leq \sqrt{1-\bra{\uState{T}}\FState{T}\ket{\uState{T}}}\\
&= \sqrt{1-\bra{\uState{T}}(\ket{\pState{T}}\bra{\pState{T}} + \RState{T})\ket{\uState{T}}}\\
&\leq \sqrt{1-|\inp{\uState{T}}{\pState{T}}|^2}
\end{align*}
where our definition of trace norm is normalized to be in $[0,1]$ and in the second inequality we used that for a (normalized) pure state $\ket{\varphi}$ and a mixed state $\rho$, we have $\trnorm{\rho - \ketbra{\varphi}{\varphi}} \leq \sqrt{1-\bra{\varphi}\rho\ket{\varphi}}$ (see, e.g., \cite[Chapter 9]{Nielsen:book}).
Therefore,
\begin{align*}
\H_{T} &= \normb {\ket{\uState{T}}-\ket{\pState{T}}}^2 \\
   &= \inp{\uState{T}}{\uState{T}} + \inp{\pState{T}}{\pState{T}} - 2 {\rm Re}( \inp{\uState{T}}{\pState{T}}) \\
   &\ge 1 - 2|\inp{\uState{T}}{\pState{T}}| > \frac{1}{10},
\end{align*}
where the next to last inequality uses the fact that $\inp{\uState{T}}{\uState{T}}=1$ and $\inp{\pState{T}}{\pState{T}}\geq 0$.
\end{proof}

The following lemma bounds the amount by which the progress measure $\H_{t}$ can increase in each step.
\begin{lemma}\label{lem:dH}
For all $k \in \{1,\ldots,N\}$ and any $0\leq t<T$,
$$ \H_{t+1} - \H_{t} \leq \frac{1-p}{p} \cdot\|\uamp{t,k}\|^2 .$$
\end{lemma}

\begin{proof}
By the definition of the progress measure,
\begin{align*}
\H_{t+1} &= \normb{\ket{\pState{t+1}} - \ket{\uState{t+1}}}^2\\
&= \normb{U_{t+1}\ket{\OpState{t+1}} - U_{t+1}\ket{\uState{t}}}^2 \\
&= \normb{\ket{\OpState{t+1}}-\ket{\uState{t}}}^2  \\
&=( \bra{\pState{t}} - 2(1-p)\bra{k,\amp{t,k}} - \bra{\uState{t}} )( \ket{\pState{t}} - 2(1-p)\ket{k,\amp{t,k}} - \ket{\uState{t}} )\\
&= \H_{t} -4(1-p)\|\amp{t,k}\|^2  +2(1-p)\inp{\amp{t,k}}{\uamp{t,k}} +2(1-p)\inp{\uamp{t,k}}{\amp{t,k}} + 4(1-p)^2\|\amp{t,k}\|^2\\
&\leq \H_{t} -4p(1-p)\|\amp{t,k}\|^2  +4(1-p)\|\amp{t,k}\|\|\uamp{t,k}\|\\
&\leq \H_{t} + \frac{1-p}{p}\|\uamp{t,k}\|^2
\end{align*}
where the last inequality follows by maximizing the quadratic expression over $\|\amp{t,k}\|$.
\end{proof}

\newtheorem*{thma}{Theorem~\ref{thm:main-thm}}
\begin{thma}
Any algorithm that solves the $p$-faulty Grover problem must use $T > \frac{p}{10(1-p)} N$ queries.
\end{thma}

\begin{proof}
By Lemma~\ref{lem:dH}, for all $k \in \{1,\ldots,N\}$,
$$\H_{T} \leq \frac{1-p}{p} \sum_{t=0}^{T-1} \|\uamp{t,k}\|^2.$$
Since for any $t$, $\ket{\uState{t}}$ is a unit vector,
$$\sum_{k=1}^{N}\H_{T}\leq \frac{1-p}{p} \sum_{k=1}^{N}\sum_{t=0}^{T-1} \|\uamp{t,k}\|^2 = \frac{1-p}{p} T.$$
To complete the proof, note that by Lemma~\ref{lem:F}, $\sum_{k=1}^{N}\H_{T}>\frac{1}{10}N$.
\end{proof}

\subsubsection*{Acknowledgments}
We thank Aram Harrow for presenting us with the faulty Grover problem and for useful discussions.
We also thank Nicolas Cerf, Fr\'ed\'eric Magniez, and the anonymous referees for useful comments.

\bibliographystyle{abbrv}
\bibliography{faulty}

\begin{thebibliography}{10}

\bibitem{Ambainis00}
A.~Ambainis.
\newblock Quantum lower bounds by quantum arguments.
\newblock In {\em Proceedings of the ACM Symposium on Theory of Computing},
  pages 636--643, New York, 2000.

\bibitem{Ambainis05}
A.~Ambainis.
\newblock Quantum search algorithms.
\newblock {\em SIGACT News}, 35(2):22--35, 2004.
\newblock quant-ph/0504012.

\bibitem{BealsBCMdW01}
R.~Beals, H.~Buhrman, R.~Cleve, M.~Mosca, and R.~de~Wolf.
\newblock Quantum lower bounds by polynomials.
\newblock {\em J. ACM}, 48(4):778--797, 2001.

\bibitem{BennettBBV97}
C.~H. Bennett, E.~Bernstein, G.~Brassard, and U.~Vazirani.
\newblock Strengths and weaknesses of quantum computing.
\newblock {\em SIAM J. Comput.}, 26(5):1510--1523, 1997.

\bibitem{BoyerBHT98}
M.~Boyer, G.~Brassard, P.~H{\o}yer, and A.~Tapp.
\newblock Tight bounds on quantum searching.
\newblock {\em Fortschritte der Physik}, 46:493--505, 1998.

\bibitem{BrassardHMT02}
G.~Brassard, P.~H{\o}yer, M.~Mosca, and A.~Tapp.
\newblock Quantum amplitude amplification and estimation.
\newblock In {\em Quantum computation and information}, volume 305 of {\em
  Contemp. Math.}, pages 53--74. Amer. Math. Soc., Providence, RI, 2002.

\bibitem{BuhrmanNRdW07}
H.~Buhrman, I.~Newman, H.~R{\"o}hrig, and R.~de~Wolf.
\newblock Robust polynomials and quantum algorithms.
\newblock {\em Theory Comput. Syst.}, 40(4):379--395, 2007.
\newblock Preliminary version in STACS 2005.

\bibitem{Grover96}
L.~K. Grover.
\newblock A fast quantum mechanical algorithm for database search.
\newblock In {\em Proceedings of the ACM Symposium on the Theory of Computing},
  pages 212--219, 1996.

\bibitem{Harrow06}
A.~Harrow.
\newblock Personal communication, 2006.

\bibitem{HoyerLS07}
P.~H{\o}yer, T.~Lee, and R.~\v{S}palek.
\newblock Negative weights make adversaries stronger.
\newblock In {\em Proceedings of the ACM Symposium on the Theory of Computing},
  pages 526--535, 2007.
\newblock quant-ph/0611054.

\bibitem{HoyerMdW03}
P.~H{\o}yer, M.~Mosca, and R.~de~Wolf.
\newblock Quantum search on bounded-error inputs.
\newblock In {\em Proceedings of ICALP 2003}, volume 2719 of {\em Lecture Notes
  in Comput. Sci.}, pages 291--299. Springer, Berlin, 2003.

\bibitem{IwamaRY05}
K.~Iwama, R.~Raymond, and S.~Yamashita.
\newblock General bounds for quantum biased oracles.
\newblock {\em IPSJ Journal}, 46(10):1234--1243, 2005.

\bibitem{KayeLM07}
P.~Kaye, R.~Laflamme, and M.~Mosca.
\newblock {\em An introduction to quantum computing}.
\newblock Oxford University Press, Oxford, 2007.

\bibitem{KnillZurekLaflamme}
E.~Knill, R.~Laflamme, and W.~H. Zurek.
\newblock Resilient quantum computation.
\newblock {\em Science}, 279(5349):342--345, 1998.

\bibitem{LongLZT00}
G.~L. Long, Y.~S. Li, W.~L. Zhang, and C.~C. Tu.
\newblock Dominant gate imperfection in {G}rover's quantum search algorithm.
\newblock {\em Physical Review A}, 61:042305, 2000.
\newblock quant-ph/9910076.

\bibitem{MagniezNRS07}
F.~Magniez, A.~Nayak, J.~Roland, and M.~Santha.
\newblock Search via quantum walk.
\newblock In {\em Proceedings of the ACM Symposium on the Theory of Computing},
  pages 575--584, New York, 2007.

\bibitem{Nielsen:book}
M.~Nielsen and I.~Chuang.
\newblock {\em Quantum Computation and Quantum Information}.
\newblock Cambridge University Press, Cambridge, UK, 2000.

\bibitem{ShapiraMB03}
D.~Shapira, S.~Mozes, and O.~Biham.
\newblock Effect of unitary noise on {G}rover\char39{}s quantum search
  algorithm.
\newblock {\em Phys. Rev. A}, 67(4):042301, 2003.

\bibitem{ShenviBW03}
N.~Shenvi, K.~R. Brown, and K.~B. Whaley.
\newblock Effects of a random noisy oracle on search algorithm complexity.
\newblock {\em Phys. Rev. A}, 68(5):052313, 2003.

\bibitem{SuzukiYNW06}
T.~Suzuki, S.~Yamashita, M.~Nakanishi, and K.~Watanabe.
\newblock Robust quantum algorithms with {$\epsilon$}-biased oracles.
\newblock In {\em Computing and combinatorics}, volume 4112 of {\em Lecture
  Notes in Comput. Sci.}, pages 116--125. Springer, Berlin, 2006.

\end{thebibliography}

\end{document}